\documentclass[12pt,twoside,letterpaper]{llncs}
\usepackage{times}
\usepackage[T1]{fontenc}
\usepackage{amsmath}
\usepackage{amssymb}
\usepackage{fancybox}
\usepackage{graphics,graphicx}
\newcommand{\junk}[1]{}

\usepackage{color}

\usepackage{color}

\begin{document}
\title{On Truly Parallel Time in Population Protocols}

\author{
Artur Czumaj
\inst{1}
\thanks{Partially supported by the Centre for Discrete Mathematics and its Applications (DIMAP) and EPSRC award EP/V01305X/1.}
\and
Andrzej Lingas
\inst{2}
\thanks{Supported by Swedish Research Council grant 621-2017-03750.}
\institute{
  Department of Computer Science and Centre for Discrete Mathematics and its Applications (DIMAP), University of Warwick, Coventry CV4 7AL, United Kingdom.
  \texttt{A.Czumaj@warwick.ac.uk}
\and
  Department of Computer Science, Lund University,
22100 Lund, Sweden.
\texttt{Andrzej.Lingas@cs.lth.se}}
}
\maketitle

\begin{abstract}
  The {\em parallel time} of a population protocol is defined as the
  average number of required interactions that an agent in the
  protocol participates, i.e., the quotient between the total number
  of interactions required by the protocol and the total number $n$ of
  agents, or just roughly the number of required rounds with $n$
  interactions.  This naming triggers an intuition that at least on
  the average a round of $n$ interactions can be implemented in $O(1)$
  parallel steps.
  We show that when the transition function of a population protocol
  is treated as a black box then
  the expected maximum number of parallel steps
  necessary to implement a round of $n$ interactions
  is $\Omega (\frac {\log n}{\log \log n})$.
 We also provide a combinatorial argument for
 a matching upper bound on the number
 of parallel steps in the average case under additional assumptions.
 \end{abstract}

\section{Introduction}
In this paper we consider the model of probabilistic population
protocols. It was originally intended to model large systems of
agents with limited resources \cite{AA}.  In this model,
the agents are prompted to interact with one another towards a
solution of a common task.  The execution of a protocol in this model
is a sequence of pairwise interactions between
agents chosen uniformly at random
\cite{AA,AR,ER}.
During an interaction, each of the two agents, called the {\em initiator} and
the {\em responder} (the asymmetry assumed in \cite{AA}), updates its
state in response to the observed state of the other agent following
the predefined (global) transition function.  The efficiency of
population protocols is typically expressed in terms of the number of
states used by agents and the number of interactions required by
solutions (e.g., with high probability (w.h.p.) or in the expectation).
There is a vast literature on population protocols, especially for
 such basic problems as majority and leader
 election \cite{AR,BGK,ER,GS20}.

 In the literature on population protocols \cite{AR,ER,GS20},
 the concept of {\em parallel time}, which is the
 number of required interactions divided by the number $n$ of agents,
 is widely spread.  In other words, one divides the sequence of
 interactions in an execution of a population protocol into
 consecutive subsequences of $n$ interactions called rounds. Then one
 estimates the expected number of required rounds or the number of
 required rounds w.h.p.

 Population protocols for any non-trivial problem require
$\Omega(n \log n)$ interactions \cite{ER}. Hence, the expressions
resulting from dividing those on the number of interactions by $n$ are
not only simpler but also more focused on the essentials.  Fast
population protocols are commonly identified with those having
poly-logarithmic parallel time.  Also, for example, when improving a
polynomial upper bound on the number of interactions to $O(n\log n)$,
one can refer to the improvement as an exponential one in terms of the
parallel time, which sounds impressive.

 Clearly, the average number of interactions that an agent takes part
 is a lower bound on the actual parallel time when the transition function
 of a population protocol is a black box. However, calling this
 trivial lower bound for parallel time may mislead readers not
 familiar with or not recalling the definition.  They may start to
 believe that by the random choice of a pair of agents for each
 interaction in a round, there should be a lot of independent
 interactions in the round that could be implemented in
 parallel. Consequently, they could believe that the whole protocol
 could be implemented in parallel in time proportional to the number
 of rounds. Unfortunately, this intuition appears too optimistic.

 It is obvious that one can construct a sequence of $n$ interactions
 that requires $n$ parallel steps when the transition function of a
 population protocol is treated as a black box.  More importantly, we
 show that the expected maximum length of a dependency chain of
 interactions in a single round of $n$ interactions is $\Theta (\frac
 {\log n}{\log \log n})$.  The lower bound implies that when the
 transition function is treated as black box and the update of the
 states of interacting agents requires $\Omega (1)$ time steps then
 the expected maximum number of parallel steps necessary to implement
 a round of $n$ interactions is $\Omega (\frac {\log n}{\log \log
   n})$.
The upper bound opens for the possibility of a
matching, fast parallel implementation of a single round of $n$ interactions in
the average case under additional assumptions.

\section{A Lower Bound on Expected Parallel Time Required by a Round}

In each round $R$ of $n$ interactions, there are $2n$ participants
slots, so on the average each agent participates in two interactions
in $R$. Consider the {\em dependency}
directed acyclic graph (DAG) $D(R)$, where
vertices correspond to interactions in the round and two vertices
$v,\ u$ are
connected by the directed edge $(v,u)$ if and only if the interaction
corresponding to $v$ precedes the interaction corresponding to $u$
and the two interactions share at least one agent.
On the average, the dependence DAG $D(R)$ has at least a
linear number of edges. They can form long directed chains excluding
the possibility of an efficient implementation of the round in
parallel.

\begin{remark}
  There is a round $R$ of $n$ interactions
  such that the dependency DAG $D(R)$
  includes a directed path of length $n-1$
  (i.e., it has depth $\ge n-1$).
  In consequence, any implementation of the round
  (when the transition function is treated as black box
  and the update of the states of interacting
  agents takes one time step) 
  requires $n$
(parallel) time steps
\end{remark}
\begin{proof}
It is sufficient to let the $i$-th agent
participate in the $i$-th and $i+1$-th interactions for $i\le n-1$.
The dependency DAG of so specified round includes a directed
path of length $n-1$.
\qed
\end{proof}

Of course, the round specified in Remark 1 yielding a dependency
path of linear length is highly unlikely.
However even in the average case, the maximal length of
a dependency path is at least almost logarithmic in $n$.

\begin{theorem}\label{theo: lower}
  The expected maximum length of a directed path  in the dependency DAG
  of a round of $n$ interactions is $\Omega (\frac {\log n}{\log \log n})$.
  Consequently, when the transition function is treated
  as black box and the update of the states of interacting
  agents requires $\Omega (1)$ time steps then
  the expected number of parallel time steps
  required to implement the round
  is $\Omega (\frac {\log n}{\log \log n})$.
\end{theorem}

\begin{proof}
Consider a sequence $S$ of $n$ pairwise interactions between the $n$ agents picked uniformly at random. We shall show that the expected maximum number of interactions in $S$ that a single agent participates is $\Omega(\frac{\log n}{\log \log n})$. To prove this we shall assume the following \emph{balls-into-bins} model. For any natural number $r$, let $[r] := \{1,\dots, r\}$. We have $2n$ balls, where for any $k \in [n]$, the balls numbered $2k-1$ and $2k$ correspond to the $k$-th interaction in $S$, and $n$ bins are in one-to-one correspondence with the $n$ agents. Allocating the balls numbered $2k-1,\ 2k$ into two distinct bins $A$ and $B$ specifies the interaction between the agents corresponding to the bins $A$ and $B$. If $A = B$ then the $k$-th interaction is not specified in this model.
Since the pairwise interactions are performed between the $n$ agents picked uniformly at random, the destinations of the balls are random. Therefore, by \cite{Go}, the expected maximum load of a bin in our model is
$\Gamma^{(-1)}(2n) - \frac32 + o(1)$, where $\Gamma$ is Euler's gamma function and it is known that $\Gamma^{(-1)}(n) = \frac{\log n}{\log \log n}(1+o(1))$. Hence, in expectation, there is an agent involved in at least $\frac{\log n}{\log \log n}(1+o(1))$ interactions. We may not exclude that the bin with the maximum load contains pairs of consecutive balls corresponding to the same interaction (which cannot be specified). However, the probability that a ball is allocated to the same bin as the previous one is only $\frac1n$. Therefore, the expected maximum load of a bin where no two balls correspond to the same interaction is still $\Omega(\frac{\log n}{\log \log n})$. Hence, the expected maximum number of interactions that the same agent participates in a round of $n$ interactions is $\Omega(\frac{\log n}{\log \log n})$.
\qed
\end{proof}

\section{An Upper Bound on Expected Maximum Length of a Dependency Chain in a Round}


  The bound in Theorem~\ref{theo: lower} follows from the fact
that one expects that at least one agent will be involved in
$\Omega(\frac{\log n}{\log \log n})$ interactions,
which immediately implies that the expected maximum length of a
directed path in the dependency DAG of a round of $n$ interactions is
$\Omega(\frac{\log n}{\log \log n})$. However, if one considers
concurrently more agents, then perhaps the expected maximum length of a
directed path in the dependency DAG can be significantly larger, that is
$\omega(\frac{\log n}{\log \log n})$? 
In this section we prove that
this is not the case, implying that the lower bound in Theorem~\ref{theo:
  lower} is asymptotically tight.

In order to derive our upper bound on the expected maximum length of a directed
path in the dependency DAG of a round consisting of $n$ interactions, we shall
identify interactions with labeled edges in $K_n$. To model directed
paths in the dependency DAG of a round of $n$ interactions, we need
the following concept.

An \emph{interference path} of length $k$ is any sequence of edges $e_1, \dots, e_k$ such that $e_i \cap e_{i+1} \ne \emptyset$ for every $1 \le i < k$.

We will consider labeled undirected multigraphs, where each edge has a unique label. We say an \emph{interference path is monotone} if the labels on the interference path form a \emph{strictly increasing sequence}.

\begin{theorem}\label{theo: upper}
Let $c$ be an arbitrary positive constant and let $n$ be a sufficiently large integer. Consider the process of selecting $n$ edges labeled $1, \dots, n$ in $K_n$ independently and uniformly at random\footnote{That is, we run the following process:
\begin{description}
\item[\quad $\triangleright$] \textbf{for $t = 1$ to $n$ do}:
\begin{itemize}
\item choose distinct $i$ and $j$ independently and uniformly at random from $[n] := \{1,\dots,n\}$;
\item assign label $t$ to edge $\{i,j\}$.
\end{itemize}
\end{description}
}.
Then, for $k = \lceil\frac{(3+c) \log n}{\log\log n}\rceil$, with probability at least $1 - \frac{1}{n^c}$, the obtained multigraph has no monotone interference path of length $k$.
\end{theorem}
\begin{proof}
  The proof is by simple counting arguments. Let $G$ be the (random) multigraph
  constructed by our process. $G$ has $n$ vertices, $n$ edges (possibly with repetitions), and each edge has a distinct label from $[n]$.

Let $\mathcal{IP}_k$ be the set of all possible labeled interference paths of length $k$ ($k \ge 1$) in $K_n$ with distinct labels in $[n]$, that is,
\begin{align*}
    \mathcal{IP}_k =
    \bigg\{\langle e_1, \dots, e_k; L \rangle:
        \forall_{1 \le i \le k}\ |e_i| = 2,
        \forall_{1 \le i \le k}\ e_i \subseteq [n],
            \\
        \forall_{1 \le i < k}\ e_i \cap e_{i+1} \ne \emptyset,
        L \subseteq [n], \text{ and }
        |L| = k
    \bigg\}
    \enspace.
\end{align*}
The meaning here is that $\langle e_1, \dots, e_k; L \rangle$ corresponds to the interference path with edges $e_1, \dots, e_k$ and with labels such that $e_i$ has label equal to the $i$-th smallest element from $L$.

Let us observe that
\begin{align}
\label{ineq:bound-size}
    |\mathcal{IP}_k| &\le
        \binom{n}{2} \cdot (2n-3)^{k-1} \cdot \binom{n}{k} \le
        n^2 \cdot (2n)^{k-1} \cdot \frac{n^k}{k!} =
        \frac{2^{k-1} \cdot n^{2k+1}}{k!}
    \enspace.
\end{align}
Indeed, we can choose any of the $\binom{n}{2}$ pairs of distinct vertices as the first edge, and then to select the $(i+1)$-st edge, we have one of the two vertices from the $i$-th edge together with one other vertex. As for the labels, they can be assigned as any subset of $[n]$ of size~$k$.

Let us take an arbitrary interference path $P = \langle e_1, \dots, e_k; L \rangle \in \mathcal{IP}_k$. Let $L = \{\chi_1, \dots, \chi_k\}$ with $\chi_i < \chi_{i+1}$ for $1 \le i < k$. For $P$ to exist in $G$, for every $1 \le i \le k$, the process must have chosen edge $e_i$ in step $\chi_i$ of the algorithm. The probability for that to happen is equal to $\frac{1}{\binom{n}{2}}$ for every $1 \le i \le k$. All the probabilities are independent for different $i$, and therefore if we let $X_P$ be the indicator random variable that $P$ is a monotone interference path in $G$, then (for $n \ge 2$)
\begin{align}
\label{ineq:bound-prob}
    \Pr[X_P = 1] &=
    \left(\frac{1}{\binom{n}{2}}\right)^k =
    \frac{2^k}{n^k (n-1)^k} \le
    \left(\frac{2}{n}\right)^{2k}
    \enspace.
\end{align}

Let $\mathcal{E}_k$ be the random event that $G$ has a monotone interference path of length $k$. By inequalities (\ref{ineq:bound-size}) and (\ref{ineq:bound-prob}), and by the union bound, we obtain the following,
\begin{align*}
    \Pr[\mathcal{E}_k] &=
    \Pr[\sum_{P \in \mathcal{IP}_k} X_P > 0] \le
    \sum_{P \in \mathcal{IP}_k} \Pr[X_P > 0] =
    |\mathcal{IP}_k| \cdot \left(\frac{1}{\binom{n}{2}}\right)^k
        \\
        & \le
    \frac{2^{k-1} \cdot n^{2k+1}}{k!} \cdot \left(\frac{2}{n}\right)^{2k}
        \le
    \frac{8^k \cdot n}{k!}
    \enspace.
\end{align*}

Finally, we use the fact that for Euler's gamma function (which for any positive integer $N$ satisfies $\Gamma(N) = (N-1)!$) we have $\Gamma^{(-1)}(N) = \frac{(1+o(1)) \cdot \log N}{\log\log N}$. Therefore, assuming $n$ is sufficiently large, if we take an arbitrary positive $c$ and in the bound above make $k \ge \frac{(c+3) \cdot \log n}{\log\log n}$ with $k = o(\log n)$, then we obtain
\begin{align*}
    \Pr[\mathcal{E}_k] & \le
    \frac{8^k \cdot n}{k!} \le
    \frac{o(n^2)}{\omega(n^{c+2})} =
    o(n^{-c})
    \enspace.
\end{align*}

Since $\mathcal{E}_k$ is the event that $G_{n,p}$ has a monotone interference path of length $k$, the bound above implies that with probability at least $1 - n^{-c}$ the random labeled graph $G_{n,p}$ has no monotone interference path of length $k = \lceil\frac{(c+3) \cdot \log n}{\log\log n}\rceil = \Theta\left(\frac{\log n}{\log\log n}\right)$.
\qed
\end{proof}

Note that monotone interference paths of length $k$ in the multigraph
in Theorem \ref{theo: upper} are in one-to-one correspondence with
directed paths of length $k$ in the dependency DAG of a round of $n$
interactions.
\junk{Hence, the expected depth of the DAG is $O(\frac {\log
  n}{\log \log n})$ and we obtain the following corollary from Theorem
\ref{theo: upper} by implementing the interactions on the same level
of the DAG in a single parallel step.
For $i = 0, 1, 2, \dots$, let the $i$-th level of the DAG  denote
  the set of its vertices (i.e., interactions) whose
  maximum distance to a source vertex (i.e., a vertex of indegree $0$)
  is $i$. It follows that the number of levels is $O(\frac {\log
  n}{\log \log n})$. Hence, we obtain the following corollary from Theorem
\ref{theo: upper} by implementing the interactions on the same level
of the DAG in $O(1)$ parallel step.

\begin{corollary}\label{cor: upper}
The expected maximum length of a directed path in the dependency DAG of a round of $n$ interactions is $O\left(\frac{\log n}{\log \log n}\right)$. Consequently, if the decomposition of the DAG into its levels is given and the update of the states of interacting agents takes $O(1)$ time steps then the expected number of parallel time steps required to implement a round of $n$ interactions is $O\left(\frac{\log n}{\log\log n}\right)$.
\end{corollary}}
Hence, we obtain the following corollary from Theorem~\ref{theo: upper}.

\begin{corollary}\label{cor: upper}
The expected maximum length of a directed path  in the dependency DAG of a round of $n$ interactions is $O\left(\frac{\log n}{\log \log n}\right)$.
\end{corollary}

For $i = 0, 1, 2, \dots$, let the $i$-th level of the DAG  denote the set of its vertices (i.e., interactions) whose maximum distance to a source vertex (i.e., a vertex of indegree 0) is $i$. It follows that the number of levels is $O(\frac{\log n}{\log \log n})$. Consequently, if the decomposition of the DAG into its levels is given and the update of the states of interacting agents takes $O(1)$ time steps, then the expected number of parallel time steps required to implement a round of $n$ interactions is $O \left( \frac{\log n}{\log\log n}\right)$.
\section{Final Remarks}

Observe that the lower bound of Theorem \ref{theo: lower} holds also with high probability, as does the upper bound of Theorem \ref{theo: upper}.

Our almost logarithmic lower bound on the expected maximum length of a
dependency chain in the dependency DAG of a round in Theorem
\ref{theo: lower} is implied by the lower bound on the expected
maximum number of interactions sharing a single agent in a round of
$n$ interactions.  It is a bit surprising that our upper bound on the
expected maximum length of a dependency chain in the DAG of the round
asymptotically matches the aforementioned lower bound.  For example,
in the round constructed in the proof of Remark 1, each agent takes
part in $O(1)$ interactions but the DAG of the round contains a
dependency chain of length $n-1$ !

The problem of  estimating the expected depth of random
circuits raised and studied by Diaz et al. in \cite{DSS94} seems
closely related. The motivation of Diaz et al. \cite{DSS94}
was an estimation of how quickly a random circuit could be
evaluated in parallel. Arya et al. improved the results of
\cite{DSS94} by providing tight $\Theta (\log n)$ bounds
on the expected depth of random circuits in \cite{AGM}.
Their improved results rely on Markov chain techniques.

One can generalize the concept of an interaction between two agents to
include that of a $k$-parallel interaction defined as a sequence of
$k$ mutually independent interactions involving $2k$ agents totally.
Then, a sequence of $t(n)$ interactions composed of $\lceil t(n)/k
\rceil $ consecutive $k$-parallel interactions can be implemented
in $O(t(n)/k)$ parallel steps.  The related problem of designing a
fast parallel randomized method of drawing $k$ disjoint pairs of
agents uniformly at random is also of interest in its own rights.  In
a recent paper \cite{BH}, Berenbrink et al. provide a method of forming
several matchings between agents in order to simulate population
protocols efficiently in parallel.

\section*{Acknowledgments}
The second author is grateful to Leszek G\k{a}sieniec, who posed the
problem studied in this paper already in 2020, for an introduction to
population protocols and to him, Jesper Jansson and Christos
Levcopoulos for some discussions.

\end{document}